\documentclass[12pt]{article}

\usepackage{amssymb,amsmath,amsthm,amscd,mathrsfs,amsbsy,amsfonts}
\usepackage{pstricks,pst-node}
\usepackage{amsbsy,young,colortbl,cite}

\usepackage{graphicx}

\usepackage{amsmath}

\usepackage{amssymb,amsthm}
\usepackage[english]{babel}
\newtheorem{lemma}{Lemma}

\newtheorem{theorem}{Theorem}

\newtheorem{example}{Example}

\newcommand{\C}{\mathbb C}

\def\({\left(}
\def\){\right)}
\def\[{\begin{eqnarray}}
\def\]{\end{eqnarray}}

\begin{document}

\title{B\"acklund transformations of $Z_n$-Sine-Gordon systems}

\author{
\  \ Xueping Yang, Chuanzhong Li\dag\footnote{Corresponding author.}\\
\small Department of Mathematics,  Ningbo University, Ningbo, 315211, China\\
\dag Email:lichuanzhong@nbu.edu.cn}

\date{}

\maketitle

\abstract{
In this paper, from the algebraic reductions from the Lie algebra $gl(n,\C)$ to its commutative subalgebra $Z_n$,  we construct the general $Z_n$-Sine-Gordon  and $Z_n$-Sinh-Gordon systems which contain many  multi-component Sine-Gordon type and Sinh-Gordon type equations. Meanwhile, we give the B\"acklund transformations of the $Z_n$-Sine-Gordon  and $Z_n$-Sinh-Gordon equations which can generate new solutions from seed solutions. To see the $Z_n$-systems clearly, we consider the $Z_2$-Sine-Gordon  and $Z_3$-Sine-Gordon equations explicitly including their B\"acklund transformations, the nonlinear superposition formula and Lax pairs.
}\\


\tableofcontents

\section {Introduction}

The Sine-Gordon equation and Sinh-Gordon equation are important integrable equations which describe many interesting phenomena including dynamics of coupled pendulums, Josephson junction arrays\cite{Josephson}, models of nonlinear excitations in complex systems in physics
and in living cellular structures, both intra-cellular
and inter-cellular\cite{Ivancevic}. The Sine-Gordon equation also accounts for continuum limits of the crystalline lattices \cite{Crystal}.

The scalar Sine-Gordon
equation originates in differential geometry and has profound
applications in physics and in life sciences which can be seen from a recent review \cite{Kundu}. Recently in \cite{Mikhailov}, a dressing method was used for the vector Sine-Gordon equation and its soliton interactions.

In \cite{zuojmp}, a new hierarchy called as a $Z_m$-KP hierarchy which take values in a maximal commutative subalgebra of $gl(m,\C)$ was constructed, meanwhile the relation between Frobenius manifolds and the dispersionless reduced $Z_m$-KP hierarchy was discussed. From the $Z_m$-KP hierarchy, one can derive some coupled equations like the coupled KdV equation which appears in \cite{Fordy,vectorkdv,kdv,Ivancevic}. We consider the Hirota quadratic equation of the commutative version of extended multi-component Toda hierarchy in \cite{EZTH} which should be useful in Frobenius manifold theory.
Because of logarithm terms, some extended Vertex operators are constructed in generalized Hirota bilinear equations which might be useful in topological field theory and Gromov-Witten theory.  Later we defined a new multi-component BKP hierarchy which takes values in a commutative subalgebra  of $gl(N,\mathbb C)$. After this, we give the gauge transformation of this commutative multi-component BKP (CMBKP) hierarchy \cite{ZBKP}. Meanwhile we construct a new constrained CMBKP hierarchy which contains some new integrable systems including coupled  KdV equations under a certain reduction. After this, the quantum torus symmetry and quantum torus constraint on the tau function of the commutative multi-component BKP hierarchy are constructed. Then a natural questions appears, i.e. what about the corresponding commutative multi-component version of other integrable systems like the well-known Sine-Gordon equation and Sinh-Gordon equation?
In this paper, we will answer this question by several steps. These steps contain their B\"acklund transformations which were studied a lot for the systems of the scalar Sine-Gordon equation \cite{SG,SHG}.

This paper is arranged as follows.  In Section 2, we recall some basic facts about the classical Sine-Gordon equation and then we give the construction of the general $Z_n$-Sine-Gordon equation, $Z_n$-Sinh-Gordon equations and their B\"acklund transformations in Section 3. What's more, we study on the B\"acklund transformations and nonlinear superposition formula of the $Z_2$-Sine-Gordon equation in section 4 and 5. In Section 6, we give the Lax equations of the $Z_2$-Sine-Gordon equation. In Section 7, we give the corresponding results on the  $Z_3$-Sine-Gordon equation.

\section{The classical Sine-Gordon equation}

To introduce the $Z_n$-Sine-Gordon equation, we firstly recall the original Sine-Gordon equation. The mathematical method of soliton equations was derived mainly through the zero curvature equation. The zero curvature equation of the original Sine-Gordon equation is as follows
\begin{equation}\label{hh}
M_{t}-N_{x}+[M,N]=0 , [M,N]=MN-NM.
\end{equation}

One can export many soliton equations by choosing appropriate  values of $M$ and $N$.  M. Ablowitz, D. J. Kaup, A. C. Newell and H. Segur  considered the zero curvature equation of the original Sine-Gordon equation with
\begin{equation}\label{ii}
M=\left(\begin{matrix}
-i\lambda&q\\
r&i\lambda\end{matrix}
\right),N=\left(\begin{matrix}
A&B\\
C&-A\end{matrix}
\right).
\end{equation}
Here the zero curvature equation \eqref{hh} becomes:
\begin{equation}\label{0o}
\begin{cases}
A_{x}=qc-rB,\\
q_{t}=B_{x}+2i\lambda B+2qA,\\
r_{t}=C_{x}-2i\lambda C-2rA.
\end{cases}
\end{equation}
For example, by choosing appropriate values of $A$, $B$, $C$ as:
\begin{equation}\label{PP}
A=\frac{g}{\lambda} , B=\frac{b}{\lambda} , C=\frac{c}{\lambda},
\end{equation}
where
\begin{equation}\label{jj}
g=\frac{i}{4}\cos{u},b=c=\frac{i}{4}\sin{u},q=r=-\frac{u_{x}}{2},
\end{equation}
 one can derive the well-known Sine-Gordon(SG) equation as following
 \begin{equation} \label{a}
u_{xt}= \sin{u}.
\end{equation}

Supposing $u$ is a solution of eq.\eqref{a}, and under the following transformation, $u'$ in the following will be another solution of eq.\eqref{a},
\begin{equation}\label{b}
(\frac{u'+u}{2})_{x}= a\sin{\frac{u'-u}{2}},
\end{equation}
\begin{equation}\label{c}
(\frac{u'-u}{2})_{t}= \frac{1}{a}\sin{\frac{u'+u}{2}}.
\end{equation}
This transformation is called the B\"acklund transformation of  SG equation. If $u$ is a solution of  SG eq.\eqref{a}, we can get another new solution $u'$ of SG equation by solving the above first-order equations (eq.\eqref{b} and eq.\eqref{c}).

\section{The $Z_n$-Sine-Gordon equation and $Z_n$-Sinh-Gordon equation}
To construct new multicomponent Sine-Gordon  and Sinh-Gordon systems which might have potential applications in biology such as DNA structural dynamics, we will consider the case when $u$ take values in a  commutative algebra $Z_n=\C[\Gamma]/(\Gamma^n)$ and $\Gamma=(\delta_{i,j+1})_{ij}\in gl(n,\C).$ This will lead to the case of $Z_n$-Sine-Gordon equation which is also equivalent to another $Z_n$-Sinh-Gordon equation. Let us firstly introduce the following lemma.
\begin{lemma}

The following identity holds
\begin{equation} \label{dddb}
\sin{\left(\begin{matrix}
a_{0}&0&0&\cdots&0\\
a_{1}&a_{0}&0&\cdots&0\\
a_{2}&a_{1}&a_{0}&\cdots&0\\
\vdots&\vdots&\vdots&\vdots&\vdots\\
a_{n}&a_{n-1}&a_{n-2}&\cdots&a_{0}\end{matrix}
\right)}=\left(\begin{matrix}
b_{0}&0&0&\cdots&0\\
b_{1}&b_{0}&0&\cdots&0\\
b_{2}&b_{1}&b_{0}&\cdots&0\\
\vdots&\vdots&\vdots&\vdots&\vdots\\
b_{n}&b_{n-1}&b_{n-2}&\cdots&b_{0}\end{matrix}
\right),
\end{equation}
where,
\begin{equation} \label{dddc}
b_{k}=\sum_{i_{1}k_{1}+i_{2}k_{2}+\cdots+i_{j}k_{j}=k}{\frac{a_{i_{1}}^{k_{1}}\cdot{a_{i_{2}}^{k_{2}}}\cdot{\cdots}\cdot{a_{i_{j}}^{k_{j}}}}{k_{1}!\cdot{k_{2}!}
\cdot{\cdots}\cdot{k_{j}!}}\sin{(a_{0}+\frac{\pi}{2}\sum _{i=1}^{j}{k_{i})}}}.
\end{equation}

\end{lemma}
\begin{proof}
In order to prove the above conclusion,  a direct calculation can lead to
\begin{equation} \label{fffm}
\sin{(a_{0}E+a_{1}\Gamma^{1}+\cdots+a_{n}\Gamma^{n})}=
(a_{0}E+a_{1}\Gamma^{1}+\cdots+a_{n}\Gamma^{n})-\frac{1}{3!}(a_{0}E
+a_{1}\Gamma^{1}+\cdots+a_{n}\Gamma^{n})^{3}+\ldots.
\end{equation}

Here we have used the results as
\begin{equation} \label{fffn}
\Gamma^{k}=(\delta_{i,j+k})_{ij},\ \ k<n; \Gamma^{n}=0_{n\times n}.
\end{equation}
Through a direct calculation, we can finish the proof by choosing the specific $k$-th diagonal position of both matrices of two  sides of the identity \eqref{dddb}.
\end{proof}

By a tedious calculation, we can get the following
$Z_n$-Sine-Gordon equation :
\begin{equation} \label{ddda}
(a_k)_{xt}=\sum_{i_{1}k_{1}+i_{2}k_{2}+\cdots+i_{j}k_{j}=k}{\frac{a_{i_{1}}^{k_{1}}\cdot{a_{i_{2}}^{k_{2}}}\cdot{\cdots}\cdot{a_{i_{j}}^{k_{j}}}}{k_{1}!\cdot{k_{2}!}
\cdot{\cdots}\cdot{k_{j}!}}\sin{(a_{0}+\frac{\pi}{2}\sum _{i=1}^{j}{k_{i})}}},\ \ 0\leq k\leq n.
\end{equation}
Also by a tedious calculation, we can derive the B\"acklund transformation of
the $Z_n$-Sine-Gordon equation in the following theorem.
\begin{theorem}\label{znbacklund}
The $Z_n$-Sine-Gordon equation \eqref{ddda}
 has the following B\"acklund transformation
\begin{equation}\label{backlund2}
\begin{cases}
(\frac{u_k'+u_k}{2})_{x}=a\sum_{i_{1}k_{1}+i_{2}k_{2}+\cdots+i_{j}k_{j}=k}{\frac{(u_{i_{1}}'-u_{i_{1}})^{k_{1}}
\cdot{(u_{i_{2}}'-u_{i_{2}})}^{k_{2}}\cdot{\cdots}\cdot{(u_{i_{j}}'-u_{i_{j}})}^{k_{j}}}{2^{k_{1}+k_{2}+\cdots+k_{j}}\cdot{k_{1}!}\cdot{k_{2}!}
\cdot{\cdots}\cdot{k_{j}!}}\sin{(\frac{u_0'-u_0}{2}+\frac{\pi}{2}\sum _{i=1}^{j}{k_{i})}}},\\
(\frac{u_k'-u_k}{2})_{t}=\frac{1}{a}\sum_{i_{1}k_{1}+i_{2}k_{2}+\cdots+i_{j}k_{j}=k}{\frac{(u_{i_{1}}'-u_{i_{1}})^{k_{1}}
\cdot{(u_{i_{2}}'-u_{i_{2}})}^{k_{2}}\cdot{\cdots}\cdot{(u_{i_{j}}'-u_{i_{j}})}^{k_{j}}}{2^{k_{1}+k_{2}+\cdots+k_{j}}\cdot{k_{1}!}\cdot{k_{2}!}
\cdot{\cdots}\cdot{k_{j}!}}\sin{(\frac{u_0'-u_0}{2}+\frac{\pi}{2}\sum _{i=1}^{j}{k_{i})}}}.
\end{cases}
\end{equation}
\end{theorem}

\begin{proof}

If $(u_k,  0\leq k\leq n)$ are solutions of the $Z_n$-Sine-Gordon equation, we  will prove that under the transformations \eqref{backlund2}, $(u'_k,  0\leq k\leq n)$ are also  solutions of the eq.\eqref{ddda}.
Here we give the proof of the above argument by direct calculation using the transformation \eqref{backlund2},
\[\notag
&&(\frac{u_k'+u_k}{2})_{xt}\\ \notag
&=&a\left(\sum_{i_{1}k_{1}+i_{2}k_{2}+\cdots+i_{j}k_{j}=k}{\frac{(u_{i_{1}}'-u_{i_{1}})^{k_{1}}
\cdot{(u_{i_{2}}'-u_{i_{2}})}^{k_{2}}\cdot{\cdots}\cdot{(u_{i_{j}}'-u_{i_{j}})}^{k_{j}}}{2^{k_{1}+k_{2}+\cdots+k_{j}}\cdot{k_{1}!}\cdot{k_{2}!}
\cdot{\cdots}\cdot{k_{j}!}}\sin{(\frac{u_0'-u_0}{2}+\frac{\pi}{2}\sum _{i=1}^{j}{k_{i})}}}\right)_{t}\\ \notag
&=&a\left(\sum_{i_{1}k_{1}+i_{2}k_{2}+\cdots+i_{j}k_{j}=k}\sum_{s=1}^j\frac{(u_{i_{1}}'-u_{i_{1}})^{k_{1}}
\cdot{\cdots}\cdot k_{s}{(u_{i_{s}}'-u_{i_{s}})}^{k_{s}-1}(u_{i_{s}}'-u_{i_{s}})_t\cdot{\cdots}\cdot{(u_{i_{j}}'-u_{i_{j}})}^{k_{j}}}{2^{k_{1}+k_{2}+\cdots+k_{j}}\cdot{k_{1}!}\cdot{k_{2}!}
\cdot{\cdots}\cdot{k_{j}!}}\right.\\
\notag && \left.\sin(\frac{u_0'-u_0}{2}+\frac{\pi}{2}\sum _{i=1}^{j}{k_{i})}\right.\\ \notag
&&\left.+\sum_{i_{1}k_{1}+i_{2}k_{2}+\cdots+i_{j}k_{j}=k}{\frac{(u_{i_{1}}'-u_{i_{1}})^{k_{1}}
\cdot{\cdots}\cdot{(u_{i_{j}}'-u_{i_{j}})}^{k_{j}}(u_0'-u_0)_t}{2^{k_{1}+k_{2}+\cdots+k_{j}+1}\cdot{k_{1}!}\cdot{k_{2}!}
\cdot{\cdots}\cdot{k_{j}!}}\cos{(\frac{u_0'-u_0}{2}+\frac{\pi}{2}\sum _{i=1}^{j}{k_{i})}}}\right)\\ \notag
&=&\frac{1}{2}\sum_{k=1}^n\sum_{i_{1}k_{1}+i_{2}k_{2}+\cdots+i_{j}k_{j}=k}{\frac{u_{i_{1}}^{k_{1}}\cdot{u_{i_{2}}^{k_{2}}}\cdot{\cdots}\cdot{u_{i_{j}}^{k_{j}}}}{k_{1}!\cdot{k_{2}!}
\cdot{\cdots}\cdot{k_{j}!}}\sin{(u_{0}+\frac{\pi}{2}\sum _{i=1}^{j}{k_{i})}}}\Gamma^k
\\ \notag
&&+\frac{1}{2}\sum_{k=1}^n\sum_{i_{1}k_{1}+i_{2}k_{2}+\cdots+i_{j}k_{j}=k}{\frac{{u'}_{i_{1}}^{k_{1}}\cdot{{u'}_{i_{2}}^{k_{2}}}\cdot{\cdots}\cdot{{u'}_{i_{j}}^{k_{j}}}}{k_{1}!\cdot{k_{2}!}
\cdot{\cdots}\cdot{k_{j}!}}\sin{(u'_{0}+\frac{\pi}{2}\sum _{i=1}^{j}{k_{i})}}}\Gamma^k.\]

Then because $u_k$ are  solutions of the $Z_n$-Sine-Gordon equation, i.e.
\[\notag
u_{kxt}=\sum_{k=1}^n\sum_{i_{1}k_{1}+i_{2}k_{2}+\cdots+i_{j}k_{j}=k}{\frac{u_{i_{1}}^{k_{1}}\cdot{u_{i_{2}}^{k_{2}}}\cdot{\cdots}\cdot{u_{i_{j}}^{k_{j}}}}{k_{1}!\cdot{k_{2}!}
\cdot{\cdots}\cdot{k_{j}!}}\sin{(u_{0}+\frac{\pi}{2}\sum _{i=1}^{j}{k_{i})}}},\]
therefore we can derive the $u_k'$ are also solutions of the $Z_n$-Sine-Gordon equation
\[\notag
u'_{kxt}=\sum_{k=1}^n\sum_{i_{1}k_{1}+i_{2}k_{2}+\cdots+i_{j}k_{j}=k}{\frac{{u'}_{i_{1}}^{k_{1}}\cdot{{u'}_{i_{2}}^{k_{2}}}\cdot{\cdots}\cdot{{u'}_{i_{j}}^{k_{j}}}}{k_{1}!\cdot{k_{2}!}
\cdot{\cdots}\cdot{k_{j}!}}\sin{(u'_{0}+\frac{\pi}{2}\sum _{i=1}^{j}{k_{i})}}}.\]
 Then we can say that the transformation \eqref{backlund2} is the B\"acklund transformation of the $Z_n$-Sine-Gordon equation.
\end{proof}
By the B\"acklund transformation of the $Z_n$-Sine-Gordon equation,
we can get other new solutions of the $Z_n$-Sine-Gordon equation. One can see it in the following example.
\begin{example}
Take the seed solution as $a_k=0,\ \ 0\leq k\leq n,$
then new solutions $a'_k$ satisfy
\end{example}
\begin{equation}\label{fffa}
a'_{kx}=2a\sum_{i_{1}k_{1}+i_{2}k_{2}+\cdots+i_{j}k_{j}=k}{\frac{a_{i_{1}}^{k_{1}}\cdot{a_{i_{2}}^{k_{2}}}\cdot{\cdots}\cdot{a_{i_{j}}^{k_{j}}}}{k_{1}!\cdot{k_{2}!}
\cdot{\cdots}\cdot{k_{j}!}}\sin{(a_{0}+\frac{\pi}{2}\sum _{i=1}^{j}{k_{i})}}},
\end{equation}
\begin{equation}\label{fffb}
a'_{kt}=\frac{2}{a}\sum_{i_{1}k_{1}+i_{2}k_{2}+\cdots+i_{j}k_{j}=k}{\frac{a_{i_{1}}^{k_{1}}\cdot{a_{i_{2}}^{k_{2}}}\cdot{\cdots}\cdot{a_{i_{j}}^{k_{j}}}}{k_{1}!\cdot{k_{2}!}
\cdot{\cdots}\cdot{k_{j}!}}\sin{(a_{0}+\frac{\pi}{2}\sum _{i=1}^{j}{k_{i})}}}.
\end{equation}
In some special cases, the results  are as follows
\\
when $n=0$, $a_{0}^{'}=4\arctan{e^{ax+a^{-1}t}}$,
\\
when $n=1$, $a_{1}^{'}=2c\frac{e^{(ax+a^{-1}t)}}{e^{(2ax+2a^{-1}t)}+1}$,
\\
when $n=2$, $a_{2}^{'}=-(\frac{1}{(e^{2ax+2a^{-1}t)}+1}+c)\frac{e^{(ax+a^{-1}t)}}{e^{(2ax+2a^{-1}t)}+1}$.

It is well known that the classical Sine-Gordon equation is equivalent to Sinh-Gordon equation.
Therefore now let us also consider the case of the  $Z_n$-Sinh-Gordon equation in similar ways as in the following lemma.
\begin{lemma}
The following identity holds
\begin{equation} \label{dddg}
\sinh{\left(\begin{matrix}
a_{0}&0&0&\cdots&0\\
a_{1}&a_{0}&0&\cdots&0\\
a_{2}&a_{1}&a_{0}&\cdots&0\\
\vdots&\vdots&\vdots&\vdots&\vdots\\
a_{n}&a_{n-1}&a_{n-2}&\cdots&a_{0}\end{matrix}
\right)}=\left(\begin{matrix}
c_{0}&0&0&\cdots&0\\
c_{1}&c_{0}&0&\cdots&0\\
c_{2}&c_{1}&c_{0}&\cdots&0\\
\vdots&\vdots&\vdots&\vdots&\vdots\\
c_{n}&c_{n-1}&c_{n-2}&\cdots&c_{0}\end{matrix}
\right),
\end{equation}
where
\begin{equation} \label{dddh}
c_{k}=\sum_{i_{1}k_{1}+i_{2}k_{2}+\cdots+i_{j}k_{j}=k}{\frac{a_{i_{1}}^{k_{1}}\cdot{a_{i_{2}}^{k_{2}}}\cdot{\cdots}
\cdot{a_{i_{j}}^{k_{j}}}}{k_{1}!\cdot{k_{2}!}\cdot{\cdots}\cdot{k_{j}!}}
\frac{e^{a_{0}}-(-1)^{\sum_{t=1}^{j}{k_{t}}}\cdot{e^{-a_{0}}}}{2}}.
\end{equation}
\end{lemma}

Also by a tedious calculation, we can derive the $Z_n$-Sinh-Gordon equation
\begin{equation} \label{ddda112}
(a_k)_{xt}=\sum_{i_{1}k_{1}+i_{2}k_{2}+\cdots+i_{j}k_{j}=k}{\frac{a_{i_{1}}^{k_{1}}\cdot{a_{i_{2}}^{k_{2}}}\cdot{\cdots}
\cdot{a_{i_{j}}^{k_{j}}}}{k_{1}!\cdot{k_{2}!}\cdot{\cdots}\cdot{k_{j}!}}
\frac{e^{a_{0}}-(-1)^{\sum_{t=1}^{j}{k_{t}}}\cdot{e^{-a_{0}}}}{2}},\ \ 0\leq k\leq n,
\end{equation}
and its B\"acklund transformation  in the following theorem.
\begin{theorem}
The $Z_n$-Sinh-Gordon equation \eqref{ddda112}
 has the following B\"acklund transformation
\begin{equation}\label{backlund3}
\begin{cases}
(\frac{u_k'+u_k}{2})_{x}=a\sum_{i_{1}k_{1}+i_{2}k_{2}+\cdots+i_{j}k_{j}=k}{\frac{(u_{i_{1}}'-u_{i_{1}})^{k_{1}}
\cdot{(u_{i_{2}}'-u_{i_{2}})}^{k_{2}}\cdot{\cdots}\cdot{(u_{i_{j}}'-u_{i_{j}})}^{k_{j}}}{2^{k_{1}+k_{2}+\cdots+k_{j}}\cdot{k_{1}!}\cdot{k_{2}!}
\cdot{\cdots}\cdot{k_{j}!}}\frac{e^{\frac{u_0'-u_0}{2}}-(-1)^{\sum_{t=1}^{j}{k_{t}}}\cdot{e^{-\frac{u_0'-u_0}{2}}}}{2}},\\
(\frac{u_k'-u_k}{2})_{t}=\frac{1}{a}\sum_{i_{1}k_{1}+i_{2}k_{2}+\cdots+i_{j}k_{j}=k}{\frac{(u_{i_{1}}'-u_{i_{1}})^{k_{1}}
\cdot{(u_{i_{2}}'-u_{i_{2}})}^{k_{2}}\cdot{\cdots}\cdot{(u_{i_{j}}'-u_{i_{j}})}^{k_{j}}}{2^{k_{1}+k_{2}+\cdots+k_{j}}\cdot{k_{1}!}\cdot{k_{2}!}
\cdot{\cdots}\cdot{k_{j}!}}\frac{e^{\frac{u_0'-u_0}{2}}-(-1)^{\sum_{t=1}^{j}{k_{t}}}\cdot{e^{-\frac{u_0'-u_0}{2}}}}{2}}.
\end{cases}
\end{equation}
\end{theorem}
The proof of this theorem is similar as the proof of the Theorem \ref{znbacklund}. We will skip it here.

 To see the $Z_n$-Sine-Gordon systems clearly, we will take take $n=2$ and $n=3$ as examples in the following several sections.

\section{ The $Z_2$-Sine-Gordon equation and its B\"acklund transformation}
In this section, we will construct the $Z_2$-Sine-Gordon equation in the  commutative algebra $Z_2=\C[\Gamma]/(\Gamma^2)$ and $\Gamma=(\delta_{i,j+1})_{ij}\in gl(2,\C).$
By a direct computation using Taylor expansion, we can get the following lemma.

\begin{lemma}

The matrix form of $\sin{u}$ and $\cos{u}$ are as following through the summation of matrices
\begin{equation} \label{l}
\sin{\left(\begin{matrix}
u_0&0\\
u_1&u_0\end{matrix}
\right)}=\left(\begin{matrix}
\sin{u_0}&0\\
u_1\cos{u_0}&\sin{u_0}\end{matrix}
\right),
\end{equation}
\begin{equation} \label{m}
\cos{\left(\begin{matrix}
u_0&0\\
u_1&u_0\end{matrix}
\right)}=\left(\begin{matrix}
\cos{u_0}&0\\
-u_1\sin{u_0}&\cos{u_0}\end{matrix}
\right).
\end{equation}
\end{lemma}

In this matrix case, after a denotation as $u_0:=u,u_1:=v,$ we can derive the following new $Z_2$-Sine-Gordon equation
\begin{equation}
\begin{cases}
u_{xt}=\sin{u},\\
v_{xt}=v\cos{u}.
\end{cases}
\end{equation}
\begin{theorem}
The $Z_2$-Sine-Gordon equation
\begin{equation}\label{n}
\begin{cases}
u_{0xt}=\sin{u_0},\\
u_{1xt}=u_1\cos{u_0},
\end{cases}
\end{equation} has the following B\"acklund transformation
\begin{equation}\label{backlund}
\begin{cases}
(\frac{u_0'+u_0}{2})_{x}=a\sin{\frac{u_0'-u_0}{2}},\\
(\frac{u_1'+u_1}{2})_{x}=\frac{u_1'-u_1}{2}\cos{\frac{u_0'-u_0}{2}},
\\
(\frac{u_0'-u_0}{2})_{t}=\frac{1}{a}
\sin{\frac{u_0'+u_0}{2}},\\
(\frac{u_1'-u_1}{2})_{t}=
\frac{1}{a}\frac{u_1+u_1}{2}\cos{\frac{u_0'+u_0}{2}}.
\end{cases}
\end{equation}
\end{theorem}
\begin{proof}

To see the B\"acklund transformation of the $Z_n$-Sine-Gordon equation clearly, here we give the proof of the above theorem by a direct calculation in terms of matrices
\[\notag
&&\left(\begin{matrix}
\frac{u_0'+u_0}{2}&0\\
\frac{u_1'+u_1}{2}&\frac{u_0'+u_0}{2}\end{matrix}
\right)_{xt}\\ \notag
&=&a\left(\begin{matrix}
\sin{\frac{u_0'-u_0}{2}}&0\\
\frac{u_1'-u_1}{2}\cos{\frac{u_0'-u_0}{2}}&\sin{\frac{u_0'-u_0}{2}}\end{matrix}
\right)_{t}\\ \notag
&=&a\left(\begin{matrix}
\frac{u_{0t}^{'}-u_{0t}}{2}\cos{\frac{u_0'-u_0}{2}}&0\\
\frac{u_{1t}^{'}-u_{1t}}{2}\cos{\frac{u_0'-u_0}{2}}-\frac{u_1'-u_1}{2}\frac{u_{0t}^{'}-u_{0t}}{2}\sin{\frac{u_0^{'}-u_0}{2}}&\frac{u_0^{'}-u_0}{2}\cos{\frac{u_{0t}^{'}-u_{0t}}{2}}\end{matrix}
\right)\\ \notag
&=&a\left(\begin{matrix}
\frac{1}{a}\cos{\frac{u_0'-u_0}{2}}\sin{\frac{u_0'+u_0}{2}}&0\\
\frac{1}{a}(\frac{u_1'+u_1}{2}\cos{\frac{u_0'+u_0}{2}}\cos{\frac{u_0'-u_0}{2}}-\frac{u_1'-u_1}{2}\sin{\frac{u_0'+u_0}{2}}\sin{\frac{u_0'-u_0}{2}})&\frac{1}{a}\cos{\frac{u_0'-u_0}{2}}\sin{\frac{u_0'+u_0}{2}}\end{matrix}
\right)\\ \notag
&=&a\left(\begin{matrix}
\frac{1}{a}\cos{\frac{u_0'-u_0}{2}}\sin{\frac{u_0'+u_0}{2}}&0\\
\frac{1}{2a}[\frac{u_1'+u_1}{2}(\cos{u_0'}+\cos{u_0})+\frac{u_1'-u_1}{2}(\cos{u_0'}-\cos{u_0})]&\frac{1}{a}\cos{\frac{u_0'-u_0}{2}}\sin{\frac{u_0'+u_0}{2}}\end{matrix}
\right)\\ \notag
&=&\left(\begin{matrix}
\frac{1}{2}(\sin{u_0'}+\sin{u_0})&0\\
\frac{1}{2}(u_1\cos{u_0}+u_1'\cos{u_0'})&\frac{1}{2}(\sin{u_0'}+\sin{u_0})\end{matrix}
\right)\\ \notag
&=&\frac{1}{2}\left(\begin{matrix}
\sin{u_0}&0\\
u_1\cos{u_0}&\sin{u_0}\end{matrix}
\right)+\frac{1}{2}\left(\begin{matrix}
\sin{u_0'}&0\\
u_1'\cos{u_0'}&\sin{u_0'}\end{matrix}
\right).\]

Then because $u_0,u_1$ are  solutions of the $Z_2$-Sine-Gordon equation, i.e.
\[\notag
u_{0xt}=\sin{u_0},\ \
u_{1xt}=u_1\cos{u_0},\]
therefore we can derive the $u_0',u_1'$ are also solutions of the $Z_2$-Sine-Gordon equation
\[\notag
u'_{0xt}=\sin{u'_0},\ \
u'_{1xt}=u'_1\cos{u'_0}.\]
 Then we can say that the transformation \eqref{backlund} is the B\"acklund transformation of the $Z_2$-Sine-Gordon equation.
\end{proof}

 Next we give a simple example with seed solutions $u_0,u_1$; and we can get another new solution of the $Z_2$-Sine-Gordon equation by solving the first-order equation of the above  B\"acklund transformation \eqref{backlund}.

\begin{example}
Let $
u_0=
u_1=0,$
then the new solutions from the B\"acklund transformation \eqref{backlund} becomes
\begin{equation}\label{t}
\begin{cases}
(\frac{{u_0}^{'}}{2})_{x}=a\sin{\frac{u_0'}{2}},\\
(\frac{{u_0}^{'}}{2})_{t}=\frac{1}{a}\sin{\frac{u_0'}{2}},
\end{cases}
\end{equation}
\end{example}

\begin{equation}\label{u}
\begin{cases}
(\frac{{u_1}^{'}}{2})_{x}=a\frac{u_1'}{2}\cos{\frac{u_0'}{2}},\\
(\frac{{u_1}^{'}}{2})_{t}=\frac{1}{a}\frac{u_1'}{2}\cos{\frac{u_0'}{2}}.
\end{cases}
\end{equation}

Here  the solutions of two sets of equations will be derived by integral calculations and we get the following conclusions:
\[\notag
{u_0}^{'}=4\arctan{e^{ax+a^{-1}t}}.\]

According to the relation of $\cos{2a}$ and $\tan{a}$ as
\[\notag
\cos{2a}=\frac{1-\tan^{2}{a}}{1+\tan^{2}{a}},\]
we derive
\[\notag
\cos{\frac{u_0'}{2}}&=&\cos[{2\arctan{e^{ax+a^{-1}t}}}]\\ \notag
&=&\frac{1-e^{2ax+2a^{-1}t}}{1+e^{2ax+2a^{-1}t}}.\]
\\
So the second equation of  equations \eqref{u} becomes:
\begin{equation}\label{u2}
\begin{cases}
(\ln\frac{{u_1}^{'}}{2})_{x}=a\frac{1-e^{2ax+2a^{-1}t}}{1+e^{2ax+2a^{-1}t}},
\\
(\ln\frac{{u_1}^{'}}{2})_{t}=\frac{1}{a}\frac{1-e^{2ax+2a^{-1}t}}{1+e^{2ax+2a^{-1}t}},\\
\end{cases}
\end{equation}

which further leads to
\begin{equation*}
u_1^{'}=2c\frac{e^{(ax+a^{-1}t)}}{e^{(2ax+2a^{-1}t)}+1}.
\end{equation*}

\subsection{The $Z_2$-Sinh-Gordon equation and its B\"acklund transformation}
In this subsection, we will use the similar method in the last section to consider the B\"acklund transformation of the $Z_2$-Sinh-Gordon equation. Basing on the well-known Sinh-Gordon equation  as following
 \begin{equation} \label{bbbb}
u_{xt}= \sinh{u},
\end{equation}
we will consider the following $Z_2$-Sinh-Gordon equation
\begin{equation}\label{sinhG}
\begin{cases}
u_{xt}=\sinh{u},\\
v_{xt}=v\cosh{u}.
\end{cases}
\end{equation}
The $Z_2$-Sinh-Gordon equation
 has the following B\"acklund transformation
\begin{equation}
\begin{cases}
(\frac{u'+u}{2})_{x}=a\sinh{\frac{u'-u}{2}},\\
(\frac{v'+v}{2})_{x}=\frac{v'-v}{2}\cosh{\frac{u'-u}{2}},
\\
(\frac{u'-u}{2})_{t}=\frac{1}{a}
\sinh{\frac{u'+u}{2}},\\
(\frac{v'-v}{2})_{t}=
\frac{1}{a}\frac{v'+v}{2}\cosh{\frac{u'+u}{2}}.
\end{cases}
\end{equation}

\section{Nonlinear superposition formula of $Z_2$-Sine-Gordon equation}

  The above B\"acklund transformation in the last section gives us an important method to derive new solutions from known solutions. But sometimes it is not easy to solve the first order equation.  Besides B\"acklund transformations,  the non-linear superposition formula is also useful to derive new solutions by algebraic calculations.
 We suppose $(c_0,c_1)$ are solutions of $Z_2$-Sine-Gordon equation with respect to a parameter  $h_{1}$ from the seed solution $(b_0,b_1)$,
 $(d_0,d_1)$ are solutions of $Z_2$-Sine-Gordon equation with respect to a parameter  $h_{2}$ from the seed solution $(b_0,b_1)$. Also we
 suppose
 $(e_0,e_1)$ are solutions of $Z_2$-Sine-Gordon equation with respect to a parameter  $h_{2}$ from the seed solution $(c_0,c_1)$
and the $(e_0,e_1)$ should also be solutions of $Z_2$-Sine-Gordon equation with respect to a parameter  $h_{1}$ from the seed solution $(d_0,d_1)$.

Then the B\"acklund transformation formula can be expressed as
\begin{equation}\label{z}
\begin{cases}
\left(
\frac{c_{0}+b_{0}}{2}
\right)_{x}=h_{1}
\sin{\frac{c_{0}-b_{0}}{2}},\\
\left(
\frac{c_{1}+b_{1}}{2}
\right)_{x}=h_{1}
\frac{c_{1}-b_{1}}{2}\cos{\frac{c_{0}-b_{0}}{2}},
\end{cases}
\end{equation}

\begin{equation}\label{aa}
\begin{cases}
\left(
\frac{e_{0}+c_{0}}{2}
\right)_{x}=h_{2}
\sin{\frac{e_{0}-c_{0}}{2}},\\
\left(
\frac{e_{1}+c_{1}}{2}
\right)_{x}=h_{2}
\frac{e_{1}-c_{1}}{2}\cos{\frac{e_{0}-c_{0}}{2}},
\end{cases}
\end{equation}

\begin{equation}\label{bb}
\begin{cases}
\left(
\frac{d_{0}+b_{0}}{2}
\right)_{x}=h_{2}
\sin{\frac{d_{0}-b_{0}}{2}},\\
\left(
\frac{d_{1}+b_{1}}{2}
\right)_{x}=h_{2}
\frac{d_{1}-b_{1}}{2}\cos{\frac{d_{0}-b_{0}}{2}},
\end{cases}
\end{equation}
\begin{equation}\label{cc}
\begin{cases}
\left(
\frac{e_{0}+d_{0}}{2}
\right)_{x}=h_{1}
\sin{\frac{e_{0}-d_{0}}{2}},\\
\left(
\frac{e_{1}+d_{1}}{2}
\right)_{x}=h_{1}
\frac{e_{1}-d_{1}}{2}\cos{\frac{e_{0}-d_{0}}{2}}.
\end{cases}
\end{equation}

By adding up eq.\eqref{z} and eq.\eqref{cc}, we can get:
\[\notag
\left(\begin{matrix}
\frac{c_{0}+b_{0}+e_{0}+d_{0}}{2}&0\\
\frac{c_{1}+b_{1}+e_{1}+d_{1}}{2}&\frac{c_{0}+b_{0}+e_{0}+d_{0}}{2}\end{matrix}
\right)_{x}
&=&h_{1}\left(\begin{matrix}
\sin{\frac{c_{0}-b_{0}}{2}}+\sin{\frac{e_{0}-d_{0}}{2}}&0\\
\frac{c_{1}-b_{1}}{2}\cos{\frac{c_{0}-b_{0}}{2}}+\frac{e_{1}-d_{1}}{2}\cos{\frac{e_{0}-d_{0}}{2}}&\sin{\frac{c_{0}-b_{0}}{2}}+\sin{\frac{e_{0}-d_{0}}{2}}\end{matrix}
\right).\]

By adding up eq.\eqref{aa} and eq.\eqref{bb}, we can get:
\[\notag
\left(\begin{matrix}
\frac{e_{0}+c_{0}+d_{0}+b_{0}}{2}&0\\
\frac{e_{1}+c_{1}+d_{1}+b_{1}}{2}&\frac{e_{0}+c_{0}+d_{0}+b_{0}}{2}\end{matrix}
\right)_{x}
&=&h_{2}\left(\begin{matrix}
\sin{\frac{e_{0}-c_{0}}{2}}+\sin{\frac{d_{0}-b_{0}}{2}}&0\\
\frac{e_{1}-c_{1}}{2}\cos{\frac{e_{0}-c_{0}}{2}}+\frac{d_{1}-b_{1}}{2}\cos{\frac{d_{0}-b_{0}}{2}}&sin{\frac{e_{0}-c_{0}}{2}}+\sin{\frac{d_{0}-b_{0}}{2}}\end{matrix}
\right).\]

According to the above two equations, we can get a system of  algebraic equations
\begin{equation}\label{cc2}
\begin{cases}
h_{1}(\sin{\frac{c_{0}-b_{0}}{2}}+\sin{\frac{e_{0}-d_{0}}{2}})=h_{2}(\sin{\frac{e_{0}-c_{0}}{2}}+
\sin{\frac{d_{0}-b_{0}}{2}}),
\\
h_{1}(\frac{c_{1}-b_{1}}{2}\cos{\frac{c_{0}-b_{0}}{2}}+\frac{e_{1}-d_{1}}{2}\cos{\frac{e_{0}-d_{0}}{2}})
=h_{2}(\frac{e_{1}-c_{1}}{2}\cos{\frac{e_{0}-c_{0}}{2}}+\frac{d_{1}-b_{1}}{2}\cos{\frac{d_{0}-b_{0}}{2}}).
\end{cases}
\end{equation}

By finishing both sides of the first formula of the eq.\eqref{cc2}, we can derive the following nonlinear superposition formula of $Z_2$-Sine-Gordon equation

\begin{equation}\label{gg}
\begin{cases}
\tan{\frac{e_{0}-b_{0}}{4}}=\frac{h_{2}+h_{1}}{h_{2}-h_{1}}\tan{\frac{c_{0}-d_{0}}{4}},\\
h_{1}(\frac{c_{1}-b_{1}}{2}\cos{\frac{c_{0}-b_{0}}{2}}+\frac{e_{1}-d_{1}}{2}
\cos{\frac{e_{0}-d_{0}}{2}})=h_{2}(\frac{e_{1}-c_{1}}{2}\cos{\frac{e_{0}-c_{0}}{2}}+\frac{d_{1}-b_{1}}{2}\cos{\frac{d_{0}-b_{0}}{2}}).
\end{cases}
\end{equation}

From the known solutions $b_{0}$,$b_{1}$,$c_{0}$,$c_{1}$,$d_{0}$ and $d_{1}$, we can get the fourth solution $e_{0},e_{1}$ by eq.\eqref{gg}
instead of solving differential equations. From the following example, one can see it clearly.
\begin{example}
From the known solutions $b_{0}=b_{1}=0$,
$c_0=4\arctan{e^{h_{1}x+h_{1}^{-1}t}}, c_1=2c\frac{e^{(h_{1}x+h_{1}^{-1}t)}}{e^{(2h_{1}x+2h_{1}^{-1}t)}+1}$,
$d_0=4\arctan{e^{h_{2}x+h_{2}^{-1}t}},d_1=
2c\frac{e^{(h_{2}x+h_{2}^{-1}t)}}{e^{(2h_{2}x+2h_{2}^{-1}t)}+1}$,
by eq.\eqref{gg} we can get the following new solution :
\begin{equation}\label{gg2}
\begin{cases}
e_0=4\arctan\left(\frac{h_{2}+h_{1}}{h_{2}-h_{1}}\frac{\operatorname{sinh}\frac{1}{2}((h_{1}+h_{2})x+(h_{1}^{-1}+h_{2}^{-1})t)}
{\operatorname{cosh}\frac{1}{2}((h_{1}+h_{2})x+(h_{1}^{-1}+h_{2}^{-1})t)}\right),
\\
e_{1}=\frac{h_{1}c_{1}\cos{\frac{c_{0}}{2}}+h_{1}(c_{1}-d_{1})\cos{\frac{e_{0}-d_{0}}{2}}-h_{2}d_{1}\cos{\frac{d_{0}}{2}}}{h_{2}\cos{\frac{e_{0}-c_{0}}{2}}}
.
\end{cases}
\end{equation}
\end{example}

\begin{proof}
Here we give the proof of the above argument by a simple calculation,
\[\notag
\tan{\frac{e_{0}-b_{0}}{4}}&=&\frac{h_{2}+h_{1}}{h_{2}-h_{1}}\tan{\frac{c_{0}-d_{0}}{4}}\\ \notag
&=&\frac{h_{2}+h_{1}}{h_{2}-h_{1}}\tan(\arctan e^{h_{1}x+{h_{1}}^{-1}t}-\arctan e^{h_{2}x+{h_{2}}^{-1}t}).\]
Using the tangent formula:
\begin{equation*}
\tan({x-y})=\frac{\tan{x}-\tan{y}}{1+\tan{x}\tan{y}},
\end{equation*}
we can get
\[\notag
e_{0}=4\arctan\left(\frac{h_{2}+h_{1}}{h_{2}-h_{1}}\frac{\operatorname{sinh}\frac{1}{2}((h_{1}
+h_{2})x+(h_{1}^{-1}+h_{2}^{-1})t)}{\operatorname{cosh}\frac{1}{2}((h_{1}+h_{2})x+(h_{1}^{-1}+h_{2}^{-1})t)}\right).
\]
Now by considering
\begin{equation*}
h_{1}(\frac{c_{1}-b_{1}}{2}\cos{\frac{c_{0}-b_{0}}{2}}+\frac{e_{1}-d_{1}}{2}
\cos{\frac{e_{0}-d_{0}}{2}})=h_{2}(\frac{e_{1}-c_{1}}{2}\cos{\frac{e_{0}-c_{0}}{2}}+
\frac{d_{1}-b_{1}}{2}\cos{\frac{d_{0}-b_{0}}{2}}),
\end{equation*}
we can derive
\begin{equation*}
e_{1}=\frac{h_{1}c_{1}\cos{\frac{c_{0}}{2}}+h_{1}(c_{1}-d_{1})\cos{\frac{e_{0}-d_{0}}{2}}
-h_{2}d_{1}\cos{\frac{d_{0}}{2}}}{h_{2}\cos{\frac{e_{0}-c_{0}}{2}}}.
\end{equation*}
\end{proof}
In this way, by the algebraic iterated operation we can get many new solutions of the $Z_2$-Sine-Gordon equation.

\section{Lax equations of the $Z_2$-Sine-Gordon equation}
In the Lax equation of the original Sine-Gordon equation, in the \eqref{ii} we suppose
\begin{equation}\label{QQ}
g=\frac{i}{4}\left(\begin{matrix}
\cos{u_0}&0\\
-u_1\sin{u_0}&\cos{u_0}\end{matrix}
\right),
\end{equation}
\begin{equation}\label{rr}
b=c=\frac{i}{4}\left(\begin{matrix}
\sin{u_0}&0\\
u_1\cos{u_0}&\sin{u_0}\end{matrix}
\right),
\end{equation}
\begin{equation}\label{ss}
q=-r=-\frac{1}{2}\left(\begin{matrix}
u_{0x}&0\\
u_{1x}&u_{0x}\end{matrix}
\right).
\end{equation}
That is to say when:
\begin{equation}\label{tt}
M=\left(\begin{matrix}
-i\lambda&q\\
r&i\lambda\\
\end{matrix}
\right)=\left(\begin{matrix}
-i\lambda&0&-\frac{1}{2}u_{0x}&0\\
0&-i\lambda&-\frac{1}{2}u_{1x}&-\frac{1}{2}u_{0x}\\
\frac{1}{2}u_{0x}&0&i\lambda&0\\
\frac{1}{2}u_{1x}&\frac{1}{2}u_{0x}&0&i\lambda\\
\end{matrix}
\right),
\end{equation}
\begin{equation}\label{uu}
N=\left(\begin{matrix}
A&B\\
C&-A\\
\end{matrix}
\right)=\left(\begin{matrix}
\frac{i}{4\lambda}\cos{u_{0}}&0&\frac{i}{4\lambda}\sin{u_{0}}&0\\
-\frac{i}{4\lambda}u_{1}\sin{u_{0}}&\frac{i}{4\lambda}\cos{u_{0}}&\frac{i}{4\lambda}u_{1}\cos{u_{0}}&\frac{i}{4\lambda}\cos{u_{0}}\\
\frac{i}{4\lambda}\sin{u_{0}}&0&-\frac{i}{4\lambda}\cos{u_{0}}&0\\
\frac{i}{4\lambda}u_{1}\cos{u_{0}}&\frac{i}{4\lambda}\sin{u_{0}}&\frac{i}{4\lambda}u_{1}\sin{u_{0}}&-\frac{i}{4\lambda}\cos{u_{0}}\\
\end{matrix}
\right),
\end{equation}
 we can export the Lax equation of the $Z_2$-Sine-Gordon equation.

\section{The $Z_3$-Sine-Gordon equation and its B\"acklund transformation}
Now let us consider the case when $u$ takes values in the  commutative algebra $Z_3=\C[\Gamma]/(\Gamma^3)$ and $\Gamma=(\delta_{i,j+1})_{ij}\in gl(3,\C)$. The SG equation is generalized to the following $Z_3$-Sine-Gordon equation

\begin{equation} \label{bbba}
\begin{cases}
z_{xt}=\sin{z},&\\
v_{xt}=v\cos{z},&\\
w_{xt}=w\cos{z}-\frac{v^{2}\sin{z}}{2},&
\end{cases}
\end{equation}
and the following lemma holds.
\begin{lemma}

The following identity holds
\begin{equation} \label{bbbc}
\sin{\left(\begin{matrix}
z&0&0\\
v&z&0\\
w&v&z\end{matrix}
\right)}=\left(\begin{matrix}
\sin{z}&0&0\\
v\cos{z}&\sin{z}&0\\
w\cos{z}-\frac{v^{2}\sin{z}}{2}&v\cos{z}&\sin{z}\end{matrix}
\right),
\end{equation}
\begin{equation} \label{bbbd}
\cos{\left(\begin{matrix}
z&0&0\\
v&z&0\\
w&v&z\end{matrix}
\right)}=\left(\begin{matrix}
\cos{z}&0&0\\
1-v\sin{z}&\cos{z}&0\\
1-w\sin{z}-\frac{v^{2}\cos{z}}{2}&1-v\sin{z}&\cos{z}\end{matrix}
\right).
\end{equation}
\end{lemma}

 Similarly, by the B\"acklund transformation of the SG equation (eq.\eqref{b} and eq.\eqref{c}),
 we can get another new solution of the $Z_3$-SG equation.

Let $
z=
v=
w=0,$
then the B\"acklund transformation will lead to the conditions of new solutions
\begin{equation}\label{bbbj}
\begin{cases}
(\frac{{z}^{'}}{2})_{x}=a\sin{\frac{z'}{2}},\\
(\frac{{z}^{'}}{2})_{t}=\frac{1}{a}\sin{\frac{z'}{2}},\\
\end{cases}
\end{equation}

\begin{equation}\label{bbbk}
\begin{cases}
(\frac{{v}^{'}}{2})_{x}=a\frac{v'}{2}\cos{\frac{z'}{2}},\\
(\frac{{v}^{'}}{2})_{t}=\frac{1}{a}\frac{v'}{2}\cos{\frac{z'}{2}},
\end{cases}
\end{equation}

\begin{equation}\label{bbbl}
\begin{cases}
(\frac{{w}^{'}}{2})_{x}=a(\frac{w'}{2}\cos{\frac{z'}{2}}-\frac{v'^{2}}{8}\sin{\frac{z'}{2}}),\\
(\frac{{w}^{'}}{2})_{t}=\frac{1}{a}(\frac{w'}{2}\cos{\frac{z'}{2}}-\frac{v'^{2}}{8}\sin{\frac{z'}{2}}),
\end{cases}
\end{equation}

Then the following  new solutions can be derived
\begin{equation}\label{bbbj}
\begin{cases}z^{'}=4\arctan{e^{ax+a^{-1}t}},\\
   v^{'}=2c\frac{e^{(ax+a^{-1}t)}}{e^{(2ax+2a^{-1}t)}+1},\\
  w^{'}=-(\frac{1}{(e^{2ax+2a^{-1}t)}+1}+c)\frac{e^{(ax+a^{-1}t)}}{e^{(2ax+2a^{-1}t)}+1}.
  \end{cases}
\end{equation}

Similar to the second order matrix, by supposing
\begin{equation}\label{bbbu}
g=\frac{i}{4}\left(\begin{matrix}
\cos{z}&0&0\\
1-v\sin{z}&\cos{z}&0\\
1-w\sin{z}-\frac{v^{2}\cos{z}}{2}&1-v\sin{z}&\cos{z}\end{matrix}
\right),
\end{equation}
\begin{equation}\label{bbbv}
b=c=\frac{i}{4}\left(\begin{matrix}
\sin{z}&0&0\\
v\cos{z}&\sin{z}&0\\
w\cos{z}-\frac{v^{2}\sin{z}}{2}&v\cos{z}&\sin{z}\end{matrix}
\right),
\end{equation}
\begin{equation}\label{ss}
q=-r=-\frac{1}{2}\left(\begin{matrix}
z_{x}&0&0\\
v_{x}&z_{x}&0\\
w_{x}&v_{x}&z_{x}\end{matrix}
\right),
\end{equation}
 one can derive the Lax equation of  the $Z_3$-Sin-Gordon equation \eqref{bbba}.

{\bf {Acknowledgements:}}
Chuanzhong Li  is  supported by the National Natural Science Foundation of China under Grant No. 11571192 and K. C. Wong Magna Fund in
Ningbo University.


\end{document}